\documentclass[12pt,letterpaper]{article}

\usepackage{xcolor}
\usepackage{amsmath,amsfonts,amsthm,amssymb}

\swapnumbers 

\newtheorem{lemma}{Lemma}[section]
\newtheorem{corollary}[lemma]{Corollary}
\newtheorem{theorem}[lemma]{Theorem}
\newtheorem{prop}[lemma]{Proposition}

\theoremstyle{definition} 
\newtheorem{definition}[lemma]{Definition}

\theoremstyle{remark}

\newtheorem{example}[lemma]{Example}

\newcommand{\<}{\langle}
\renewcommand{\>}{\rangle}

\newcommand{\Ker}{\operatorname{Ker}}
\renewcommand{\Im}{\operatorname{Im}}
\begin{document}

\title{Spectral order unit spaces and JB-algebras}

    \author{Anna Jen\v cov\'a and
 Sylvia Pulmannov\'{a}{\footnote{ Mathematical Institute, Slovak Academy of
Sciences, \v Stef\'anikova 49, SK-814 73 Bratislava, Slovakia;
jenca@mat.savba.sk, pulmann@mat.savba.sk. }}}

\date{}

\maketitle

\begin{abstract}
Order unit spaces with comparability and spectrality properties as introduced by Foulis are studied.
We define continuous functional calculus for order unit spaces with the comparability
property and Borel functional calculus for spectral order unit spaces. Applying the
conditions of Alfsen and Schultz, we characterize order unit spaces with comparability
property
that are JB-algebras. We also prove a  characterization of Rickart JB-algebras as those
JB-algebras
for which every maximal associative subalgebra is monotone $\sigma$-complete, extending an
analogous 
result of Sait\^o and Wright for C*-algebras.

\end{abstract}
\medskip


\section{Introduction} \label{sc:Intro}

We continue the study of spectrality of  order unit spaces in the sense
introduced by Foulis (\cite{ Fcomgroup, Fcompog, Fgc, Funig, Forc, FPmonot, FP}), started
in \cite{JP1}, where  we compared
this approach with the well known notion of spectral duality of Alfsen and Schultz
\cite{AS1, AS2}. It was proved that the purely algebraic approach of Foulis is strictly
more general than that of Alfsen and Schultz which is based on the geometry of dual pairs of order
unit and base normed spaces. This was illustrated on the examples of JB-algebras and of
order unit spaces constructed from Banach spaces. Moreover, the structure of spectral
order unit spaces was described in detail.

Spectrality in the sense of Foulis requires the existence of a distinguished family of
positive  idempotent mappings,  called a compression base. A compression base is
introduced as a generalization of the set of compressions on a von Neumann algebra
$\mathcal A$, defined by a projection $p\in \mathcal A$ as
\[
a\mapsto pap,\qquad a\in\mathcal A.
\]
Spectrality of the order unit space is then characterized by two properties of the chosen
compression base: the
comparability property (that can be interpreted as existence of a unique orthogonal
decomposition for  each element into its positive and negative parts) and the projection
cover property (generalizing the existence of supports of self-adjoint elements in
$\mathcal A$).

In the present paper, we show that for an order unit space $A$  with the comparability
property, we can define a continuous functional calculus, which can be extended to Borel
functional calculus if $A$ is spectral. Further, we characterize the order unit spaces
with comparability property that are JB-algebras. For order unit spaces in spectral
duality, Alfsen and Shultz proved several equivalent conditions for being JB-algebras.
We show that those conditions, formulated in terms of compressions and by functional
calculus, can be used also in our more general
situation. This result is applied to the example of order unit spaces obtained from
Banach spaces (which we here call generalized spin factors) and it is proved that such an
order unit space is a JB-algebra if and only if it is a spin factor (that is, the underlying Banach space is a Hilbert
space).

In the last section, we concentrate on spectral JB-algebras. It was proved in \cite{JP1}
that a JB-algebra is (Foulis) spectral if and only if it is Rickart (recall that a
JB-algebra is Alfsen-Schultz spectral if and only if it is a JBW-algebra). Rickart
JB-algebras were introduced as a generalization of Rickart C*-algebras \cite{AyArz, Arz}. Using recent
results in \cite{Wet2}, we prove for JB-algebras an
analogue of the following characterization by Sait\^o and Wright \cite{SW}: A C*-algebra
is Rickart if and only if every  maximal commutative self-adjoint subalgebra is monotone
$\sigma$-complete.

\section{Spectrality in order unit spaces}

Recall that an \emph{order unit space} is an archimedean partially ordered real vector space with a distinguished order unit. An order unit space will be denoted by a triple
$(A, A^+,1)$, where $A^+:=\{ a\in A: 0\leq a\}$ is the positive cone and $1\in A^+$ denotes the order unit. The unit interval in $A$ will be denoted by $E:=\{ e\in A: 0\leq e\leq 1\}$.
Notice that $E$ is an effect algebra in the sense of \cite{FoBe}.
The \emph{order unit norm}  $\|. \|_1$ on $A$ is defined by
\[
\| a\|_1:= \inf \{ \lambda \in {\mathbb R}^+: -\lambda \leq a \leq \lambda \}, \  a\in A.
\]

The theory of spectral duality in order unit spaces was developed by Alfsen and Shultz, \cite{AS1, AS2}. In this paper, we will deal with
another approach to spectrality developed in \cite{FP}, based on the works  by Foulis \cite{Fcomgroup, Fcompog, Fgc, Funig, FPmonot, Forc}.
A comparison of these two approaches can be found in \cite{JP1}. In the sequel we briefly describe some details of the latter approach.

\begin{definition}\label{de:comprous} A positive linear mapping $J: A\to A$ is a \emph{compression with focus } $p$ on $A$ if for all $e\in E$,
\begin{enumerate}
\item[(F1)] $J(1)=p\in E$ (that is, $J$ is normalized),
\item[(F2)] $e\leq p \ \implies \ J(e)=e$,
\item[(F3))] $J(e)=0\ \implies \ e\leq 1-p$.
\end{enumerate}
If $J$ satisfies (F1) and (F2) but not necessarily  (F3), we say that $J$ is a  \emph{retraction}.
\end{definition}

Note that the maps defined above were called \emph{F-compressions} in \cite{JP1}, in order to
distinguish them from the stronger notion of compressions by Alfsen and Schultz. Since we
will deal with only one type of maps, we will call them compressions throughout this work.

Let us denote
\begin{equation}
\Ker^+(J):=\{ a\in A^+: J(a)=0\},\quad  \Im^+(J):=\{ a\in A^+: J(a)=a\}.
\end{equation}

The compressions $J$ and $J'$ are \emph{complementary} if $\Ker^+(J)=\Im^+(J')$ and $\Ker^+(J')=\Im^+(J)$.
Let $J(1)=p$ and $J'(1)=p'$. It turns out that $J$ and $J'$ are complementary iff $p'=1-p$, \cite[Lemma 3.6]{JP1}.

Recall that an element  $p\in E$ is \emph{sharp} if $p\wedge(1-p)=0$ (i.e., if $x\leq p, x\leq 1-p$ then $x=0$), and an element $p\in E$ is \emph{principal} if $e,f\leq p$ and
$e+f\leq 1$, then $e+f\leq p$. It is easy to see that a principal element is sharp.

\begin{lemma}\label{le:foc} {\rm  \cite[Lemma 3.7]{JP1}}. Let $p$ be the focus of a retraction $J$. Then $p$ is principal.
\end{lemma}

Elements $a,b$ in an effect algebra $E$ are \emph{Mackey compatible} (denoted $a\leftrightarrow b$) in $E$, if there are elements $c,a_1,b_1\in E$ such that $c+a_1+b_1\in E$ and $a=c+a_1, b=c+b_1$.
If $F$ is a subset of $E$, we say that $a,b\in F$ are Mackey compatible in $F$, if they are compatible in $E$, and  elements $c,a_1,b_1$ exist in $F$.

Let $E$ be an effect algebra and let $P$ be a subalgebra of $E$. Then $P$ is a \emph{normal subalgebra} if for all $d,e,f\in E$ such that
$d+e+f\leq 1$ and $d+e, d+f\in P$ we have $d\in P$. Notice that if $P$ is normal, then elements $p,q$ in $P$ are Mackey compatible in $E$ if and only if they are Mackey compatible in $P$.

\begin{definition}\label{de:comprba} A \emph{compression base} for $A$ is a family $(J_p)_{p\in P}$ of compressions on $A$ , indexed by their own foci,
such that $P$ is a normal subalgebra of $E$ and whenever $p,q,r\in P$ and $p+q+r\leq 1$, then $J_{p+r}\circ J_{q+r}=J_r$. Elements of $P$ will be called \emph{projections}.
\end{definition}

\begin{example}\label{ex:cstar}
An important example is the order unit space $A=B_{sa}(\mathcal H)$ of bounded self-adjoint operators on a Hilbert space
$\mathcal H$, or, more generally, the self-adjoint part $A=\mathcal A_{sa}$ of a
C*-algebra $\mathcal A$, with the set of
maps $U_p: a\mapsto pap$ for a projection $p\in P(\mathcal H)$. It can be seen that  each
$U_p$ is a compression with focus $p$ and $(U_p)_{p\in
P(\mathcal H)}$ is a compression base. Moreover, any compression on $A$ has this form
for some projection $p$, \cite{Fcompog}.

\end{example}

In the sequel, we fix a compression base $(J_p)_{p\in P}$ for $A$. Notice that since $P$ is a subalgebra, any $J_p$ has a (fixed) complementary compression $J_{1-p}$.
Using the fact that all projections are principal elements, it is easily seen, that if $p,q\in P$ and $p+q\leq 1$, then $p+q=p\vee q$, so that $P$ with the orthocomplementation
$p\mapsto 1-p$ is an \emph{orthomodular poset} (OMP) \cite{PtPu, Harding}. It follows that if $p,q\in P$ are Mackey compatible, then $p\wedge q$ and $p\vee q$ exist in $P$.
It was shown that $P$ is moreover a {\emph regular} OMP, which means that any pairwise Mackey compatible subset in $P$ is contained in a Boolean subalgebra of $P$ \cite{Harding}. A maximal subset of pairwise
Mackey compatible elements is a maximal Boolean subalgebra in $P$ and is called a \emph{block} of $P$. The set $P$ is covered by its blocks.

The following notion for $p\in P$ and $a\in A$ was introduced in  \cite[Def. 1.5]{FP}.

\begin{definition}\label{decompat} For $p\in P$ and $a\in A$, we say that $a$ \emph{ commutes} with $p$ if
\[
a=J_p(a)+J_{1-p}(a).
\]
\end{definition}

The set of all $a$ commuting with $p\in P$ will be denoted by $C(p)$ and called the \emph{commutant} of $p$.
For a subset $Q\subseteq P$, we denote $C(Q):=\bigcap_{p\in Q} C(p)$.
If $B\subseteq P$ is a block of $P$, the set $C(B)$ is called a \emph{C-block} of $A$.

The following facts are easily checked (see also \cite[Lemma 1.3]{FP}).

\begin{lemma}\label{le:commut} Let $a\in A$, $p\in P$. Then
\begin{enumerate}
\item[{\rm(i)}] If $J_p(a)\leq a$, then $a\in C(p)$.
\item[{\rm(ii)}] If $a\in A^+$, then $a\in C(p)$ if and only if $J_p(a)\leq a$.
\item[{\rm(iii)}] If $a\in E$, then $a\in C(p)$ if and only if $a$ and $p$ are Mackey
compatible. In this case, $J_p(a)=a\wedge p$.
\item[{\rm(iv)}] If $q\in P$, then $q\in C(p)$ if and only if $p\in C(q)$ if and only if
\[
J_pJ_q=J_qJ_p=J_{p\wedge q}
\]
\end{enumerate}
\end{lemma}

We denote the set of all projections commuting with $a\in A$ by $PC(a)$ and for $B\subseteq A$, $PC(B):=\bigcap_{a\in B}PC(a)$.

The following subset was called a \emph{bicommutant} of $a$ in $P$  in \cite{JP1}. To
avoid confusion with other notions to be introduced below,  we prefer here to call it  a \emph{P-bicommutant} of $a$:
\[
P(a)=PC(PC(a)\cup \{a\}).
\]

\begin{definition}\label{de:projcov} We say that the compression base $(J_p)_{p\in P}$  has the  \emph{  projection cover property}
if for every effect $e\in E$ there is an element $p\in P$ such that for $q\in P$, we have
$e\leq q$ iff $p\leq q$. Such an element is unique and is called
the \emph{ projection cover} of $e$, denoted by $e^o$.
\end{definition}

It turns out that for every $e\in E$, $e^o\in P(e)$, \cite[Lemma 3.16]{JP1}. If $(J_p)_{p\in P}$ has the projection cover property, then $P$ is an orthomodular lattice (OML),
\cite[Theorem 6.4]{Fgc}, \cite[Theorem 3.17]{JP1}.

\begin{definition}\label{de;compar} {\rm \cite[Def. 1.6]{FP}} We say that the compression base $(J_p)_{p\in P}$ in $A$ has the \newline \emph{ comparability property} if, for every $a\in A$,
$P^{\pm}(a) \neq \emptyset$, where
\[
P^{\pm}(a):=\{ p\in P(a) \ \mbox{and}\ J_{1-p}(a)\leq 0\leq J_p(a)\}.
\]
\end{definition}

 Notice that if a compression base $(J_p)_{p\in P}$ has the comparability property,
then $p\in P$ iff $p$ is sharp \cite[Lemma 3.20]{JP1}.

Let $p\in P^\pm(a)$ and put $b:=J_p(a)$, $c=-J_{1-p}(a)$, then we have
\begin{equation}\label{eq:pm}
a=b-c,\ b,c\in A^+, \ J_p(b)=b, J_p(c)=0.
\end{equation}
Any decomposition of $a$ of the form (\ref{eq:pm}) for some $p\in P$ is called a \emph{$P$-orthogonal decomposition} of $a$.

By \cite[Proposition 3.19]{JP1}, if $A$ has the comparability property, then every element $a\in A$ admits a unique orthogonal decomposition denoted by $a=a^+- a^-$,
$a^+=J_p(a)$, $a^-=- J_{1-p}(a)$. Moreover, we have $a^+, a^-, |a|:= a^++a^- \in C(PC(a))$.

\begin{theorem}\label{th:cor3.26} {\rm \cite[Corollary 3.26]{JP1}} Assume that $(J_p)_{p\in P}$ has the comparability property.
\begin{enumerate}
\item For any $a\in A$ and $p\in P$, $a\in C(p)$ if and only if $P(a)\subseteq C(p)$.

\item For any $a\in A$, there is some block $B\subseteq P$ such that $a\in C(B)$.

\item  For any block $B$, $C(B)=\overline{\mathrm{span}}(B)$.
\end{enumerate}

\end{theorem}

The statement (ii) above shows that if $A$ has the comparability property, then it is
covered by C-blocks. We next describe the structure of the C-blocks in more detail.

\begin{lemma}\label{le:C(B)}{ \rm \cite[Lemma 3.24]{JP1}} Assume that $(J_p)_{p\in P}$ has the comparability property and let $B$ be a block of $P$. Then $C(B)$ is an order unit space
and $(\bar{J_p})_{p\in B}$ where ${\bar J_p}:=J_p|C(B)$, is a compression base in $C(B)$ with the comparability property.
\end{lemma}

\begin{theorem}\label{th:block} {\rm \cite[Theorem 3.25]{JP1}} Let $(J_p)_{p\in P}$ be a compression base in $A$ with the comparability property. Let $B\subseteq P$ be a block.
Then there is a totally disconnected compact Hausdorff space $X$  such that
\begin{enumerate}
\item $B$ is isomorphic (as a Boolean algebra) to the Boolean algebra ${\mathcal P}(X)$ of all clopen subsets in $X$.
\item $C(B)$ is isomorphic (as an order unit space) to a norm-dense order unit subspace in $C(X,{\mathbb R})$.
\item If $A$ is norm-complete, then $C(B)\simeq C(X,{\mathbb R})$ (as an order unit space).
\end{enumerate}
\end{theorem}

In the comparability case, we can extend the notion of commutativity to all pairs of elements in $A$.

\begin{definition}\label{de:extcompat} {\rm \cite[Definition 3.26]{JP1}} Let $(J_p)_{p\in P}$ be a compression base in $A$ with the comparability property. We say that $a,b\in A$
\emph{commute}, in notation $aCb$,  if $p$ and $q$ are compatible for all $p\in P(a)$ and $q\in P(b)$.
\end{definition}

By \cite[Proposition 3.27]{JP1}, two elements in $A$ commute if and only if they are in the same C-block.  In addition, the C-blocks of $A$ are precisely the maximal sets of mutually commuting elements.

Now we introduce the notion of spectrality of order unit spaces in the sense of \cite{FP}.

\begin{definition}\label{de:spect} {\rm \cite[Definition 1.7]{FP}} The compression base $(J_p)_{p\in P}$ in an order unit space is \emph{ spectral} if it has both the projection cover
and the comparability property. 
\end{definition}

\begin{theorem}\label{th:equiv} {\rm \cite[Theorem 3.33]{JP1}} Assume that $A$ is norm complete and let $(J_p)_{p\in P}$ be a compression base with the  comparability
property. The following are equivalent:
\begin{enumerate}
\item $(J_p)_{p\in P}$ is spectral.

\item For any block $B$ of $P$, $({\bar J_p}:=J_p|C(B))_{p\in B}$ is a spectral compression base in $C(B)$.

\item  Any C-block in $A$ is monotone $\sigma$-complete.

\item  Any C-block in $A$ is isomorphic to $C(X,{\mathbb R})$ for some basically disconnected compact Hausdorff space $X$.

\item  $P$ is monotone $\sigma$-complete.
\end{enumerate}

\end{theorem}

If $(J_p)_{p\in P}$ is a spectral compression base in $A$, then there is a unique mapping $*:A\to P$, called the \emph{Rickart mapping}, such that for all $a\in A$ and $p\in P$,
\begin{equation}\label{eq:rickart}
p\leq a^* \ \Leftrightarrow \ a\in C(p
) \ \mbox{and}\ J_p(a)=0.
\end{equation}
If a compression base  has the  comparability property, then the projection cover property
is equivalent to the existence of the Rickart mapping \cite[Thm 2.1]{FP}. In particular, for $a\in E$ we have $a^*=1-a^o$.
Moreover, by \cite[Proposition 3.35]{JP1}, for any $a\in A$, $(a^+)^{**}$ is the least element in $P^{\pm}(a)$.

We may now define the \emph{ spectral resolution} of $a$ as the family  $(p_{a,\lambda})_{\lambda \in {\mathbb R}}\subseteq P$ defined by \cite[Def. 2]{FP}
\[
p_{a,\lambda}:=((a-\lambda)^+)^*.
\]
We have the following characterization of compatibility via spectral resolution:

\begin{lemma}\label{le:spr} {\rm \cite[Cor 3.36]{JP1}} If $p\in P$, then $a\in C(p) \Leftrightarrow p_{a,\lambda}\in C(p)$ for all $\lambda \in {\mathbb R}$.
\end{lemma}

Moreover,  every element in $A$ can be written as a norm-convergent integral of
Riemann-Stieltjes type \cite[Rem. 3.1]{FP}
\[
a=\int^{U_a}_{L_a-0}\lambda dp(\lambda)
\]
where the \emph{spectral lower and upper bounds} for $a$ are defined by $L_a:=\sup\{ \lambda \in {\mathbb R}: \lambda \leq a\}$ and
$U_a:=\inf \{ \lambda \in {\mathbb R}: a\leq \lambda \}$, respectively.

\begin{definition} We say that an order unit space is spectral (has the comparability
property) if there exists a compression base on $A$ that is spectral (has the comparability
property).

\end{definition}

  Note that it follows from \cite[Prop. 4.7 and Thm. 5.10]{JP2} that if $A$ is spectral (has the
comparability property), then any compression base
such that $P$ is the set of all sharp elements must be spectral (have the comparability
property). This can be seen from the fact that both the comparability and the projection
cover property are obtained from the behaviour of the restrictions of the compression base
to C-blocks. The C-blocks in $A$ only depend on the blocks in
$P$ (Theorem \ref{th:cor3.26}), moreover, any compression with focus $p\in B$ for a block
$B\subseteq P$ restricts to the unique compression with focus $p$ on $C(B)$, since $C(B)$
is isomorphic to a dense subspace $C(X)$ (Theorem \ref{th:block}) which is the self-adjoint part of an  abelian C*-algebra (cf.
Example \ref{ex:cstar}).
 By the results of \cite{JP2}, the
spectral resolution of an element  $a\in A$ coincides with the (unique) spectral
resolution  in any C-block containing $a$, so that the spectral resolution does not depend
on the choice of the spectral compression base.

\section{Functional calculus in order unit spaces} \label{sec:fc}

In this section, we introduce functional calculus in an order unit Banach space with the
comparability or spectrality property.

\subsection{Commutants}

Assume that the compression base $(J_p)_{p\in P}$ in $A$ has the comparability property.
 Let us denote
(cf. Definition \ref{de:extcompat})
\[
C(a):=\{ b\in A: aCb\}.
\]
The set $C(a)$ will be called the \emph{commutant} of $A$. For a subset $A_0\subset A$, we denote $C(A_0)=\bigcap_{a\in A_0}C(a)$.
Note that for $A_0\subseteq  P$,  we get the same notion of a commutant as before.

\begin{prop}\label{prop:commutant} Let the compression base $(J_p)_{p\in P}$ in $A$ have the comparability property. Let $a\in A$ and let $A_0\subseteq A$. Then
\begin{enumerate}
\item $C(a)=C(P(a))$.
\item $a\in C(a)$.
\item For any $b\in C(A_0)$, we have $P(b)\subseteq C(A_0)$.
\item  $C(A_0)$ is a norm-closed subspace  containing 1
 and $\{J_p|_{C(A_0)}\}_{p\in P\cap C(A_0)}$ is a compression base in $C(A_0)$ with the comparability property.
\end{enumerate}
\end{prop}

\begin{proof} (i)  By definition and Theorem \ref{th:cor3.26},
\[
C(a)=\{b:\ aCb\}=\{b:\ P(b)\subseteq C(P(a))\}=\{b:\ b\in C(P(a))\}=C(P(a)).
\]
This also implies (ii)  since  by the definition of $P(a)$, we have $a\in C(P(a))$.
The statement (iii) follows from the definition of the commutant: let $a\in A_0$, then $bCa$ means that $a\in C(P(b))$ so that $P(b)\subseteq C(a)$.
To prove (iv), note that for $p\in P$, $C(p)$ is clearly a norm-closed subspace and $1\in C(p)$. Since $C(a)=\cap_{p\in P(a)} C(p)$ by (i), the same holds true for $C(a)$ and clearly also for $C(A_0)$.
Let $F$ be the interval $[0,1]$ in $C(A_0)$, we show that $Q:=C(A_0)\cap P$ is a normal subalgebra in $F$. So let $d,e,f\in F$ such that $d+e+f\le 1$ and $d+e,d+f\in Q$. This implies that $d\in P\cap C(A_0)=Q$, since $P$ is a normal subalgebra in $E=A[0,1]$. Let $q\in Q$ and $b\in C(A_0)$. Then for any $a\in A_0$ and $p\in P(a)$ we have
$q,b\in C(p)$ and therefore
\[
J_q(b)=J_q(J_p(b)+J_{1-p}(b))= J_p(J_q(b))+J_{1-p}(J_q(b))
\]
so that $J_q(b)\in C(p)$ for all $p\in P(a)$, hence $J_q(b)\in C(P(a))=C(a)$. This implies that all the compressions $J_q$, $q\in Q$ preserve $C(A_0)$, so they define
compressions on $C(A_0)$. It is clear that $\{J_q\}_{q\in Q}$ is a compression base. The comparability property follows by comparability in $A$ and the fact that by (iii),
$P^\pm(b)\subseteq P(b)\subseteq Q$ for any $b\in C(A_0)$.

\end{proof}

We will also denote by $CC(A_0)=C(C(A_0))$ the \emph{bicommutant} of  $A_0\subset A$.

\begin{prop}\label{prop:bicomutant}
\begin{enumerate}
\item If $A_0\subseteq A$ is commutative then $CC(A_0)$ is commutative.
\item For $a\in A$, $P(a)=CC(a)\cap P$.
\item For $a\in A$, $CC(a)=\overline{\mathrm{span}}(P(a))$ (the norm-closed linear span of $P(a)$).
\end{enumerate}

\end{prop}

\begin{proof} Since $A_0$ is commutative, we have $A_0\subseteq C(A_0)$ so that $CC(A_0)\subseteq C(A_0)=C(CC(A_0))$, which implies that $CC(A_0)$ is commutative.
For (ii), we clearly have $P(a)\subseteq CC(P(a))=CC(a)$ by Prop. \ref{prop:commutant} (i). Conversely, let $p\in CC(a)\cap P$, then $a\in C(p)$ since $a\in C(a)$. Moreover, we have
$PC(a)=C(a)\cap P$, so that $p\in PC(PC(a)\cup \{a\})=P(a)$.
For (iii), note that by Theorem \cite[Thm. 3.22]{JP1},   we have for any $b\in A$ that $b\in \overline{\mathrm{span}}(P(b))$. Further, we obtain from (ii) and
Prop. \ref{prop:commutant} (iii) that for $b\in CC(a)$
\[
P(b)\subseteq CC(a)\cap P=P(a),
\]
so that  $CC(a)\subseteq \overline{\mathrm{span}}(P(a))$. The other inclusion follows from (ii) and Prop. \ref{prop:commutant} (iv).

\end{proof}

The following theorem is similar to \cite[Thm. 3.25]{JP1} ( Theorem \ref{th:block}), the proof is exactly the same, using the fact that $P(a)$ is a Boolean algebra.

\begin{theorem}\label{thm:bicommutant} Let $A$ have the comparability property and let  $a\in A$. There is a totally disconnected compact Hausdorff space $X$ such that
\begin{enumerate}
\item $P(a)$ is isomorphic (as a Boolean algebra) to the Boolean algebra $\mathcal P(X)$ of all clopen subsets in $X$.
\item $CC(a)$ is isomorphic (as an order unit space) to a norm-dense order unit subspace in $C(X,\mathbb R)$.
\item If $A$ is norm-complete, then $CC(a)\cong C(X,\mathbb R)$ (as an order unit space).
\end{enumerate}

\end{theorem}

\subsection{Functional calculus}

We can now define a continuous functional calculus  in the case when $A$ has the comparability property and is norm-complete, so it is an order unit Banach space.
In what follows, for a (compact Hausdorff) topological space $X$  we will use the notation $C(X)=C(X,\mathbb R)$, since we will work with real functions only.

For $a\in A$, let $\Phi:CC(a)\to C(X)$ be the isomorphism as in Theorem \ref{thm:bicommutant} (iii) and let $f=\Phi(a)\in C(X)$ be the element corresponding to $a$.
 Since $C(X)$ is a commutative C*-algebra, there is a continuous functions calculus in
 $C(X)$: let $sp(a):=sp(f)$ be the spectrum of $f$ (that is, the range of $f$ in $\mathbb
 R$). Then $sp(a)$ is a compact subset of $\mathbb R$ and the subalgebra $A(f)$ of $C(X)$
 generated by $f$ and $1$ is isomorphic to $C(sp(a))$. In this isomorphism, $f$ is
 associated with the identity function $t\mapsto t$ and any element in $A(f)$ is of the
 form $g\circ f$ for some $g\in C(sp(a))$. The set $sp(a)$ will be called the spectrum of
 $a$ and the subalgebra $A(f)$ is isomorphic to a subspace in $CC(a)\subseteq A$,
 isomorphic to $C(sp(a))$. The element corresponding to $g\in C(sp(a))$ in this
 isomorphism will be denoted by $g(a):=\Phi^{-1}(g\circ f)$.

 Assume that $A$ is spectral. Our next aim is to show that in this case we may define a
 Borel functional calculus.

 \begin{prop}\label{prop:cca_spectral} Let $A$ be spectral, $a\in A$. Then $CC(a)$ is a spectral order unit space.

 \end{prop}

\begin{proof} We already know  that $CC(a)$ is an order unit space with the comparability property  with respect to the restricted compression base $\{J_p\}_{p\in P(a)}$
(cf. Proposition \ref{prop:commutant}).
Let $b\in CC(a)$, then $b^\circ\in P(b)\subseteq  CC(a)$. It follows that $CC(a)$ has
also the projection cover property and is therefore spectral.

\end{proof}

 We now have by Thm. \ref{thm:bicommutant}, Prop. \ref{prop:cca_spectral} and
 \cite{FPmonot}, that
$X$ is basically disconnected, equivalently,  $C(X)$ is monotone $\sigma$-complete. The
next result shows that in this case the bounded Baire functions are `close' to continuous
functions. This fact was used in the proof of the Loomis-Sikorski theorem for $\sigma$-MV algebras, \cite{Dv, DvPu}.
 We give the proof in our setting below for convenience of the reader.
 Below,  $\mathcal B_0(X)$ denotes the set of all bounded Baire functions on $X$.

\begin{prop}\label{prop:ls_theorem} Let $X$ be a  basically disconnected compact Hausdorff space. There is a mapping $\Theta$ of $\mathcal B_0(X)$ onto $C(X)$ such that

 \begin{enumerate}
 \item for $h\in \mathcal B_0(X)$, $\{x\in X:\ h(x)\neq \Theta(h)(x)\}$ is a meager set,
 \item $\Theta$ preserves the products and the linear structure: $\Theta(h_1h_2)=\Theta(h_1)\Theta(h_2)$, $\Theta(\lambda h_1+\mu h_2)=\lambda\Theta(h_1)+\mu\Theta(h_2)$,
 \item $\Theta$ is positive: $\Theta(h)\ge 0$ if $h\ge 0$,
 \item $\Theta(\max\{h_1,h_2\})=\Theta(h_1)\vee \Theta(h_2)$, for $h_1,h_2\in \mathcal B_0(X)$,
 \item if $h_n\in \mathcal B_0(X)$ is a nondecreasing sequence with pointwise supremum $h$, then $\bigvee_n\Theta(h_n) =\Theta(h)$ in $C(X)$.
 \end{enumerate}
\end{prop}

\begin{proof}
Since $\mathcal B_0(X)$ is the set of functions measurable with respect to the $\sigma$-algebra of Baire sets, it is
closed under pointwise suprema and pointwise limits. We will first show that it is the smallest  algebra of
 bounded functions with this property, containing $C(X)$.

 So let this smallest algebra of functions be denoted by $\mathcal T$. Clearly, $\mathcal T\subseteq \mathcal B_0(X)$. Let
 \[
 \mathcal T_0:=\{A\subset X:\ \chi_A\in \mathcal T\}.
 \]
 Then $\mathcal T_0$ is a $\sigma$-algebra of subsets of $X$. Indeed, clearly $\chi_X=1_X\in C(X)\subseteq \mathcal T$, so that $X\in \mathcal T_0$, and if $\chi_A\in \mathcal T$
 then also $\chi_{X\setminus A}=1_X-\chi_A\in \mathcal T$, so that $\mathcal T_0$ is closed under complements. Let $A_n\in \mathcal T_0$, then
 \[
 \chi_{\cup_n A_n}(x)=
 \min\{1, \sum_n \chi_{A_n}(x)\}=-\max\{-1,\sum_n -\chi_{A_n}\}
 \]
 so that $\chi_{\cup_n A_n}\in \mathcal T$, hence $\cup_n A_n \in \mathcal T_0$.

 We next show that $\mathcal T$ is the set of all bounded measurable functions with respect to $\mathcal T_0$. So let $f\in \mathcal T$, we have to prove that for all $t\in \mathbb R$,
 $F_t:=\{x:\ f(x)\ge t\}\in \mathcal T_0$. Since $\mathcal T$ is a linear subspace containing the constant $1_X$, and $1_X$ is clearly $\mathcal T_0$-measurable, we may assume that $f$ has values in $[0,1]$. So let $t\le 0$, then $F_t=X\in \mathcal T_0$. Similarly, if $t>1$, then $F_t=\emptyset\in \mathcal T_0$. For $t=1$, we have $F_1=\{x: f(x)=1\}$ and
  \[
  \chi_{F_1}=1-\lim_{n}\min\{n(1-f),1\}\in \mathcal T
  \]
  We now use the fact that the binary fractions of the form $\sum_{i=1}^n a_i2^{-i}$, $a_i\in \{0,1\}$ are dense in $(0,1)$ and we have $F_{t}=\cap_n F_{r_n}$ for $r_n\nearrow
   t$, so that is enough to assume that $t$ is a binary fraction. So let $t=\sum_{i=1}^n a_i2^{-i}$, we will proceed by induction on $n$. So assume that $n=1$, then we have for $t=1/2$
   \[
   F_{1/2}=\{x:\ \min\{2f(x),1\}=1\}\in \mathcal T_0
   \]
 by the previous considerations,  since $\min\{2f, 1_X\}\in \mathcal T$ has values in $[0,1]$. Assume the assertion holds for $n$ and let $t=\sum_{i=1}^{n+1}a_i 2^{-i}$. Put
 $s=\sum_{i=2}^{n+1} a_i2^{1-i}$.
 If  $a_1=0$, then $t=s/2$ and we have
 \[
 F_t=\{x:\ 2f(x)\ge s\}=\{x: \min\{2f(x),1\}\ge s\}\in \mathcal T_0,
 \]
by the induction assumption and the fact that $\min\{2f,1_X\}\in \mathcal T$ has values in
$[0,1]$. Similarly, if $a_1=1$, then $t=(1+s)/2$ and
\[
F_t=\{x:\ 2f(x)-1\ge s\}=\{x: \max\{2f(x)-1,0\}\ge s\}\in \mathcal T_0.
\]
Hence all functions in $\mathcal T$ are $\mathcal T_0$-measurable. Conversely, since any $\mathcal T_0$-measurable function is a pointwise limit of simple functions which are
linear combinations of characteristic functions in $\mathcal T$, it is clear that any such function belongs to $\mathcal T$.
Since  $C(X)\subseteq \mathcal T$, all continuous functions are $\mathcal T_0$-measurable, and
 since Baire $\sigma$-algebra is smallest with this property, it must be contained in $\mathcal T_0$. This implies that $\mathcal B_0(X)\subseteq \mathcal T\subseteq \mathcal B_0(X)$
(cf. the proof of \cite[Thm. 7.1.7]{DvPu}).

We now introduce an equivalence relation in $\mathcal B_0(X)$ as follows: $f\sim g$ if $\{x:\ f(x)\ne g(x)\}$ is a meager set. Note that $f\sim g$ implies $f=g$ for $f,g\in C(X)$, so that
 every equivalence class may contain a unique continuous function (if any). Let
\[
\mathcal T':=\{f\in \mathcal B_0(X):\ f\sim g,\text{ for some } g\in C(X)\}.
\]
We will show that $\mathcal T'=\mathcal B_0(X)$. It is easily seen that $\mathcal T'$ is a subalgebra in $\mathcal B_0(X)$ and $\max\{f,g\}\in \mathcal T'$ for $f,g\in \mathcal T'$.
It is now enough to prove that $\mathcal T'$ is  closed under pointwise limits of non-decreasing sequences, since then $\mathcal T'$ is also closed under all pointwise limits
 and this implies $\mathcal T'=\mathcal B_0(X)$ by minimality of $\mathcal B_0(X)$ in the first part of the proof.

So let $f_n\in \mathcal T'$, $f_n\le f_{n+1}$, $f_n \nearrow f$  and let $b_n\in C(X)$ be such that $f_n\sim b_n$. Then  $b_n\vee b_{n+1} \sim f_n\vee f_{n+1}=f_{n+1}\sim b_{n+1}$, so that $b_n\vee b_{n+1}$
 and $b_{n+1}$ are continuous functions that differ only on a meager set, so that we must
 have $b_n\vee b_{n+1}=b_{n+1}$ and $b_n\le b_{n+1}$. Further, since $f_n$ is bounded,
 $b_n$ is a bounded increasing sequence.  Since $X$ is basically disconnected,
  $C(X)$ is monotone $\sigma$-complete, so that there is some $b=\bigvee_n b_n$ in $C(X)$. Let us also denote $b_0:=\lim_n b_n$, the pointwise limit. Our aim is to prove that
  $f\sim b_0$ and $b_0\sim b$, so that $f\in \mathcal T'$.
  We have
  \[
  \{x:\  f(x)\ne b_0(x)\}=\{x:\ f(x)< b_0(x)\}\cup \{x:\ f(x)>b_0(x)\}.
  \]
  If $f(x)<b_0(x)=\lim_n b_n(x)$, then there is some $n$ such that $f_n(x)\le f(x)<b_n(x)\le b_0(x)$. Similarly, if $f(x)>b_0(x)$ then $f(x)\ge f_n(x)>b_0(x)\ge b_n(x)$. It follows that
  \[
  \{x:\  f(x)\ne b_0(x)\}\subseteq \bigcup_n \{x:\ f_n(x)\ne b_n(x)\},
  \]
 which is a meager set. Hence $f\sim b_0$. To prove $b_0\sim b$, we invoke the construction from \cite[Lemma 9.1]{Good}: put
 \[
 \tilde b_0(x):=\inf_{U\in \mathcal N(x)}\sup_{y\in U} b_0(y),
 \]
 where $\mathcal N(x)$ is the collection of open neighborhoods of $x$ in $X$, then $\tilde b_0$ is continuous if $b_0^{-1}(\alpha,\infty)$ is an open $F_\sigma$ subset for all
 $\alpha\in \mathbb R$. This last condition is satisfied, since
 \[
 b_0^{-1}(\alpha,\infty)=\{x:\ b_0(x)>\alpha\}=\{x:\ \exists n, b_n(x)>\alpha\}=\cup_n b_n^{-1}(\alpha,\infty)
 \]
is an open $F_\sigma$-set, since $b_n$ are continuous. We also have $\tilde b_0(x)=b_0(x)$ in all $x$ where $b_0$ is continuous, so that
\[
\{x:\ \tilde b_0(x)\ne b_0(x)\}\subseteq \{x:\ b_0 \text{ is discontinuous in } x\}.
\]
Since $b_0$ is a pointwise limit of a sequence of continuous functions, this last set is meager, hence $f\sim b_0\sim \tilde b_0$ and $f\in \mathcal T'$. Note also that by definition, $\tilde b_0\ge b_0\ge b_n$, for all $n$,
 so that $b_0\le b=\bigvee_n b_n\le \tilde b_0$. This implies $b\sim b_0\sim \tilde b_0$, so that $b=\tilde b_0$.

 It follows that $\mathcal T'=\mathcal B_0(X)$, so that every bounded Baire function is equivalent to a unique continuous
 function. Let $\Theta: \mathcal B_0(X)\to C(X)$
 be the corresponding map. Then it is easily seen by construction that $\Theta$ has the properties (i)-(v)
(cf. the proof of \cite[Thm. 7.1.22]{DvPu}).

\end{proof}

 We can now define a Borel function calculus on $A$. For a subset $S\subseteq \mathbb R$,
 we will denote by $\mathcal B(S)$ the set of bounded Borel functions $S\to \mathbb R$.

\begin{theorem}[Borel functional calculus]\label{thm:borel_fc} Let $A$ be a spectral order unit Banach
space. Then for every $a\in A$, there is  a positive unital linear map $\Psi_a: \mathcal B(sp(a))\to CC(a)\subseteq
A$ such that:
\begin{enumerate}
\item $\Psi_a(id)=a$, where $id$ is the identity function $t\mapsto t$,
\item $\Psi_a(\chi_B)\in P(a)\subseteq P$ for the characteristic function $\chi_B$ of any Borel subset $B\subseteq sp(a)$,
\item $\Psi_a(\max\{h_1,h_2\})=\Psi_a(h_1)\vee \Psi_a(h_2)$ in $CC(a)$,
\item If $h_n$ is a nondecreasing sequence of functions in $\mathcal B(sp(a))$ with a
pointwise supremum $h$, then $\bigvee_n \Psi_a(h_n)=\Psi(h)$ in $CC(a)$.

\end{enumerate}
We will use the notation $g(a):=\Psi_a(g)$ for any $g\in \mathcal B(sp(a))$.
\end{theorem}

Let us stress that the suprema in the above theorem are taken in $CC(a)$ and they are not
necessarily suprema in $A$.

\begin{proof} As before, there is an isomorphism $\Phi :CC(a)\to C(X)$ for some basically disconnected compact
Hausdorff space  $X$. Let $f=\Phi(a)$, so that   $sp(a)=sp(f)$.
Notice that for  $g\in \mathcal B(sp(a))$, $g\circ f$ is in $\mathcal B_0(X)$, so we may   put
\[
\Psi_a(g):=\tilde \Phi(g\circ f)\in CC(a),
\]
where
\[
\tilde \Phi:=\Phi^{-1}\circ \Theta: \mathcal B_0(X)\to C(X)\to CC(a).
\]
 It follows by Proposition \ref{prop:ls_theorem} (ii) and
(iii) that $\Psi_a$ is positive and linear, further, we have
\[
\Psi_a(1_{sp(a)})=\tilde \Phi(1_{sp(f)}\circ f)=\tilde \Phi(1_X)=1,
\]
since $\Theta(1_X)=1_X$ and $\Phi$ is unital.  Let us check the properties (i) -
(iv). First, we have
\[
\Psi_a(id)=\tilde \Phi(f)=a,
\]
since clearly $\Theta (f)=f$. Further, let $\chi_B$ be the characteristic function of a Borel
subset $B\subseteq sp(a)$, then
\[
\Psi_a(\chi_B)=\tilde \Phi(\chi_B\circ f)=\tilde \Phi(\chi_{f^{-1}(B)}).
\]
By Proposition \ref{prop:ls_theorem} (ii),
$\Theta(\chi_{f^{-1}(B)})^2=\Theta(\chi_{f^{-1}(B)}^2)=\Theta(\chi_{f^{-1}(B)})$, which means
that $\Theta(\chi_{f^{-1}(B)})$ is a continuous functions with values in $\{0,1\}$, that is, a
characteristic function of a clopen subset. By construction, such functions in $C(X)$
correspond to elements in $P(a)$. This proves (ii). The property (iii) follows easily from
Proposition \ref{prop:ls_theorem} (iv).
Let now $h_n\in \mathcal B(sp(a))$ be a nondecreasing sequence
with pointwise supremum $h$, then $h_n\circ f$ is a nondecreasing sequence in $\mathcal
B_0(X)$ with pointwise supremum $h\circ f$ and we have by Proposition
\ref{prop:ls_theorem} (v) that
\[
\Psi_a(h)=\tilde \Phi(h\circ f)=\bigvee_n \tilde \Phi(h_n\circ f)=\bigvee_n \Psi_a(h_n).
\]

\end{proof}

We now use the above Borel functional calculus to describe the spectral resolution of $a$.
Let us define
\[
\xi_a(B):=\chi_B(a),\qquad B\subseteq sp(a) \text{ a Borel subset.}
\]
Then by Theorem \ref{thm:borel_fc}, it is easily seen that $\xi_a$ is a
$\sigma$-homomorphism from the Borel subsets in $sp(a)$ into the Boolean $\sigma$-algebra
$P(a)$. Further, if $g\in \mathcal B(sp(a))$, then it is not difficult to see that $\xi_{g(a)}=\xi_a\circ
g^{-1}$.

Let $h\in C(X)$ and let $\{x\in X, h(x)\ne 0\}^-$ be the support of $h$. Since $X$ is
basically disconnected, the support of any continuous function is a clopen set. It is
easily seen that the characteristic function of the support of $h$ is the projection cover
$h^\circ$ and for $b\in CC(a)$, $\Phi(b^\circ)=\Phi(b)^\circ$.  Let now $\lambda\in
\mathbb R$ and let $f=\Phi(a)$, then
\[
p_{a,\lambda}=1-(a-\lambda)_+^\circ=\Phi^{-1}(1-(f-\lambda)_+^\circ)=\Phi^{-1}(1-\chi_{f^{-1}((\lambda,\infty))^-})=\Phi^{-1}(\chi_{f^{-1}((-\infty,\lambda])^i})
\]
where for $Y\subseteq X$, $Y^i$ denotes the interior of $Y$.
Now note that the difference of the closed set $f^{-1}((-\infty,\lambda])$ and  its
interior is a
meager set, hence we obtain
\[
\xi_a((-\infty,\lambda])=\tilde\Phi(\chi_{f^{-1}((-\infty,\lambda])})=\Phi^{-1}(\chi_{f^{-1}((-\infty,\lambda])^i})=
p_{a,\lambda}.
\]

Recall that $a$ has an integral expression with respect to its spectral resolution:
\[
a=\int \lambda dp_{a,\lambda},
\]
in the sense that the integral sums $\sum_{i=1}^n
\lambda_i(p_{a,\lambda_i}-p_{a,\lambda_{i-1}})$ for $\lambda_0\le L_a\le \lambda_1\le
\dots\le U_a\le \lambda_n$ converge to $a$ in norm as $\max_i|\lambda_i-\lambda_{i-1}|\to
0$. If $g\in C(sp(a))$, then we have
\begin{equation}\label{eq:gint}
g(a)=\int g(\lambda)dp_{a,\lambda}.
\end{equation}
Indeed, let $u_1,\dots, u_n$ be mutually orthogonal projections in $P(a)$, (that is, the
corresponding clopen subsets of $X$ are mutually disjoint) and let $b= \sum_i \lambda_i
u_i$, such elements are called simple. Then it is easy to see that $b^k=\sum_i
\lambda_i^ku_i$ for any $k\in \mathbb N$ and hence $p(b)=\sum_i
p(\lambda_i)u_i$ for any polynomial $p$. By Stone-Weierstrass theorem and continuity of
the functional calculus, we obtain that $g(b)=\sum_ig(\lambda_i)u_i$ for any $g\in
C(sp(a))$. Since $a$ is the norm limit of integral sums $b_n$ that are simple, it follows that
$g(a)$ is the norm limit of $g(b_n)$. Hence $g(a)$ is the norm limit of  integral sums of the form $\sum_i
g(\lambda_i)(p_{a,\lambda_i}-p_{a,\lambda_{i-1}})$  as $\max_i
|\lambda_i-\lambda_{i-1}|\to 0$, this implies the statement.

For a Borel function $g:sp(a)\to \mathbb R$ we do not have the above integral as a norm
limit in general, but only in some weaker sense:
For any $\sigma$-additive state $\rho$ on $CC(a)$, $\rho \circ\xi_a$ defines a probability
measure on Borel subsets of $sp(a)$ (and can be extended to $\mathbb R$), with
distribution function $F_a(\lambda)=\rho(p_{a,\lambda})$. Using the integral representation of $a$ we obtain
\[
\rho(a)=\int \lambda dF_a(\lambda)=\int \lambda \rho\circ\xi_a(d\lambda).
\]
For a Borel measurable function  $g:sp(a)\to {\mathbb R}$ we have $\rho\circ
\xi_{g(a)}=(\rho\circ \xi_a)\circ g^{-1}$
and so
\[
\rho(g(a))=\int \lambda(\rho\circ \xi_a)(g^{-1}(d\lambda)) = \int g(\lambda)(\rho\circ \xi_a)(d\lambda)
\]
by the integral transformation theorem. If there are enough  $\sigma$-additive states on
$CC(a)$ (or on $A$) to separate points, then we can define the integral \eqref{eq:gint} in
a weak sense. In general, we only have the relation of the spectral resolutions as
\[
\xi_{g(a)}=\xi_a\circ g^{-1}\ \implies \
p_{g(a),\lambda}=\xi_a(g^{-1}(-\infty,\lambda])).
\]

 \section{Order unit spaces that are JB-algebras}

We recall that a \emph{Jordan algebra} over $\mathbb R$ is a vector space $A$ over $\mathbb R$ equipped with a commutative bilinear product $\circ$ that satisfies the identity
\begin{equation}\label{eq:jordprod}
(a^2\circ b)\circ a=a^2\circ(b\circ a) \ \mbox{for all} \ a,b\in A,
\end{equation}
here $a^2=a\circ a$.
A \emph{JB-algebra} is a Jordan algebra $A$ over $\mathbb R$ with identity element $1$
equipped with a complete norm satisfying the following requirements for $a,b\in A$
(\cite{AS1}):
\begin{enumerate}
\item[{\rm(JBi)}] $\|a\circ b\|\leq \|a\|\|b\|$,
\item[{\rm(JBii)}] $\|a^2\|\leq \|a\|^2$,
\item[{\rm(JBiii)}] $\|a^2\|\leq \|a^2+b^2\|$.
\end{enumerate}
Notice that by (JBi) the Jordan product is norm-continuous. A JB-algebra that is the dual
of a Banach space is called a JBW-algebra. For more information see e.g. \cite{HOS}. We will
use \cite{AS1} as a general reference.

The following theorem shows relations between JB-algebras and order unit spaces.

\begin{theorem}{\rm \cite[Theorem 1.11]{AS1}} If $A$ is a JB-algebra, then $A$ with its given
norm  is an order unit
Banach space with $A^+=\{a^2: a\in A\}$ and the identity $1$ as distinguished order unit. Furthermore, for each $a\in A$
\begin{equation}\label{eq:JBous}
-1\leq a\leq 1 \ \implies \ 0\leq a^2\leq 1.
\end{equation}
Conversely, if $A$ is a order unit Banach space equipped with a Jordan product for which
the distinguished order unit is the identity and satisfies  (\ref{eq:JBous}), then $A$ is a JB-algebra with the order unit norm.
\end{theorem}

In any Jordan algebra,  a \emph{triple product} $\{abc\}$  is defined by
\begin{equation}
\{abc\}= (a\circ b)\circ c+(b\circ c)\circ a-(a\circ c)\circ b.
\end{equation}
The special case of the triple product $\{abc\}$ with $a=c$ is denoted
\begin{equation}\label{tripl} U_ab=\{aba\}=2a\circ(a\circ b)-a^2\circ b.
\end{equation}
We will also need the  Jordan multiplication operator
\[
T_a: A\to A, \quad  T_ab=a\circ b,\ b\in A.
\]
It is clear that $T_a$ defines a bounded linear operator on $A$ for all $a\in A$.

An element $p$ in a JB-algebra $A$ is called a \emph{projection} if $p^2=p$. In this case,
we have
\begin{equation}\label{eq:UpTp}
U_p=2T_p^2-T_p,\qquad T_p=\frac12(I+U_p-U_{1-p}),
\end{equation}
here $I$ is the identity map $A\to A$.

Let $P$ be
the set of all projections in $A$.  It was shown in \cite{JP1} that every compression on a JB-algebra $A$ is of the form
\begin{equation}\label{JBcompr}
J_p(a)=U_p a=\{pap\}
\end{equation}
for some $p\in P$ and the  set  $(U_p)_{p\in P}$ forms a compression base in $A$ \cite[Corollary 5.11]{JP1}.
Note that this is a unique compression base in $A$ which is maximal, that is, is not
contained in any other compression base. In what follows, we always assume this
compression base in a JB-algebra  $A$.

In this section, we will assume that $A$ is an order unit Banach space with a compression
base $(J_p)_{p\in P}$ satisfying the comparability property.
We will show some additional conditions under which $A$ becomes a JB-algebra. These
conditions are the same as in \cite{AS1} and the proofs follow the same steps, we only
show that our (weaker) assumption of comparability is sufficient.

\begin{theorem}\label{th:JBalg1}
 $A$ is a JB-algebra if and only if for all $p,q\in P$, we have
\begin{equation}\label{eq:JB_condition}
J_p(q)+J_{1-p}(1-q)=J_q(p)+J_{1-q}(1-p).
\end{equation}
\end{theorem}

\begin{proof} (\cite[Theorem 9.43]{AS1}) For $p\in P$, we define the operator $T_p:A\to A$ by $T_p=\frac12(I+J_p-J_{1-p})$. The condition \eqref{eq:JB_condition}
 is equivalent to
 \begin{equation}\label{eq:JB_conditionT}
 T_p(q)=T_q(p).
 \end{equation}
 Note that in the case that $A$ is a JB algebra,
$T_p$ is precisely the multiplication operator, so that the conditions
\eqref{eq:JB_conditionT} are necessary.

For the converse, let $A_0=\mathrm{span}(P)$, then $A_0$ is a subspace in $A$, which is
dense in $A$ by comparability (c.f. Theorem  \ref{th:cor3.26}). Let $a,b\in A_0$, $a=\sum_i \alpha_i p_i$ and $b=\sum_j\beta_jq_j$ and define
\[
a\circ b=\sum_i \alpha_i T_{p_i}(b)=\sum_{i,j} \alpha_i\beta_j T_{p_i}(q_j)=\sum_{i,j} \alpha_i\beta_j T_{q_j}(p_i)=\sum_j \beta_j T_{q_j}(a)=b\circ a.
\]
It is clear that this does not depend on the representation of $a$ and $b$ as linear combinations of projections. Moreover, $a\circ b$ is bilinear, commutative and separately norm continuous. Since $A_0$ is dense in $A$, $a\circ b$ uniquely extends to a bilinear, commutative and separately continuous product on $A$.
To show that $A$ is a JB-algebra, it is now enough to show that $a\circ b$ is a Jordan product satisfying the condition
\begin{equation}\label{eq:jordan_JB}
-1\le a\le 1 \quad \implies \quad 0\le a\circ a\le 1,\qquad a\in A,
\end{equation}
cf. \cite[Theorem 1.11]{AS1}.

  By Theorem \ref{thm:bicommutant}, we have $a\in CC(a)=\overline{\mathrm{span}}(P(a))\simeq C(X)$
  for some totally disconnected compact Hausdorff space $X$. Let
  $\Phi: CC(a)\to C(X)$ be the isomorphism. Note that for $p,q\in P(a)$, we have $p\circ q= T_p(q)=J_p(q)=p\wedge q$, so that we have
  \[
  \Phi(p\circ q)=\Phi(p\wedge q)=\Phi(p)\wedge \Phi(q)=\Phi(p)\Phi(q),
  \]
  since $\Phi(p)$ and $\Phi(q)$ are characteristic functions. For a simple element $a_0=\sum_i \alpha_i p_i$ where $p_i\in P(a)$, we
  obtain
  \[
  a_0\circ a_0=\sum_{i,j} \alpha_i\alpha_j J_{p_i}(p_j)=\sum_{i,j}\alpha_i\alpha_jp_i\wedge p_j,
  \]
  which implies $\Phi(a_0\circ a_0)= \Phi(a_0)^2$.
  Since $a$ is the norm limit of simple elements of this form, we obtain by continuity of $\circ$ and the product of functions that $\Phi(a\circ a)=\Phi(a)^2$.
  Since $\Phi(a)$ is a function and the isomorphism preserves the order, this implies \eqref{eq:jordan_JB}.

To finish the proof, we now have to prove the Jordan identity
\[
a^2\circ(b\circ a)=(a^2\circ b)\circ a,
\]
note that we may now safely write $a\circ a=a^2$ since by the previous paragraph it is in agreement with the continuous functional calculus obtained in the previous section.
Note first that for commuting  projections $p,q\in P$ we have (using Lemma \ref{le:commut}
(iii))
\[
(p\circ b)\circ q=T_qT_p(b)=T_pT_q(b)= p\circ (b\circ q).
\]
Let now  $a_0=\sum_i \alpha_i p_i$, where $p_1,\dots,p_n$ are mutually orthogonal projections (i.e. $p_i\wedge p_j=0$).  Then $a^2=\sum_i \alpha^2_i p_i$ and we have
\[
a^2\circ (b\circ a)=\sum_{i,j} \alpha_i^2 \alpha_jp_i\circ (b\circ p_j)=\sum_{i,j} \alpha_i^2\alpha_j (p_i\circ b)\circ p_j=(a^2\circ b)\circ a.
\]
Now note that any simple element of the form $\sum_i \beta_i q_i$ with $q_i\in P(a)$ can be written as a finite linear combination of mutually orthogonal projections in $P(a)$.
 Since each $a\in A$ is a norm limit of such simple elements, the statement follow by continuity of the product.

 \end{proof}

 There is another characterization of JB-algebras  using the functional calculus on the space $A$. Define the product
 \[
 a\circ b= \frac 14[ (a+b)^2+(a-b)^2],\qquad a,b\in A.
 \]
\begin{theorem} The product  $a\circ b$ is bilinear if and only if \eqref{eq:JB_condition}
holds. In this case, $A$ is a JB-algebra and  $a\circ
b$ is the Jordan product in $A$.

\end{theorem}

 \begin{proof} (cf. \cite[Cor. 9.44]{AS1})
 Assume that \eqref{eq:JB_condition} holds. By Theorem \ref{th:JBalg1} and its proof, $A$
 is then a JB-algebra with Jordan product $*$ such that $a*a=a^2$ agrees with the continuous functional calculus on $A$.
 By bilinearity of $*$,
 \[
 a\circ b= \frac14[ (a+b)*(a+b)+(a-b)*(a-b)]=a*b
 \]
is bilinear as well.

Conversely, assume that $a\circ b$ is bilinear. It is clear by the definition that $a\circ b$ is commutative and  continuous in both variables.
We will use the fact that a JB-algebra can be characterized as a  complete order unit space with a commutative power associative product for which the order unit 1 is the
identity and \eqref{eq:jordan_JB} holds, \cite[Thm. 2.49 and A.51]{AS1}.  It is clear that  $a\circ a=a^2$ coincides with the squares defined by the continuous functional calculus on $A$, so that
\eqref{eq:jordan_JB} holds as before, since it holds in $C(X)\simeq CC(a)$. Similarly, $\circ$ is power associative since so is the usual product of functions in
$C(X)$.
It follows that $A$ with respect to this product is a JB-algebra. Note that since $A$ has comparability, $P$ must contain all the sharp elements \cite[Lemma 3.20]{JP1},
By \cite[Theorem 5.10]{JP1}, for each sharp element $p$ there is a unique retraction with focus $p$, namely $U_p=\{p\cdot p\}$. It follows that for all $p\in P$, $J_p=U_p$ for which
the condition \eqref{eq:JB_condition} holds.

 \end{proof}

\subsection{Generalized spin factors}

We now consider the following example of an order unit space obtained from a normed space
$(X,\|\cdot\|)$, cf. \cite{JP1, Berd, BerdOd}. In this case, $A=\mathbb R\times X^*$, with the
positive cone
\[
A^+=\{(\alpha,y),\ \|y\|\le \alpha\}
\]
here $(X^*,\|\cdot\|)$ is the dual space to $(X,\|\cdot\|)$, and the order unit is
$(1,0)$. By \cite[Thm. 6.5]{JP1}, there exists a spectral compression base $(J_p)_{p\in P}$
in $A$ if and only if $X$ is a reflexive and smooth Banach space. This means
that $X$ is reflexive and every nonzero element $x\in X$ attains its norm at a unique
element of the dual unit ball that is,
\[
\partial_x:=\{y\in X^*:\ \|y\|=1, y(x)=\|x\|\}
\]
is a singleton for all $0\ne x\in X$. Equivalently, $X$ is reflexive and the dual space $X^*$ is strictly
convex, which means that every boundary point of the unit ball is an extreme point.
By \cite{BerdOd}, see also  \cite[Thm. 6.6]{JP1}, $A$ is spectral in the Alfsen-Shultz
sense if and only if $X$ is reflexive, smooth,  and strictly convex as well. In this case,
the space $A$ was called a generalized spin factor in \cite{BerdOd}. In the present work,
 by a generalized spin factor we mean any space $A$ obtained in the above way from a Banach
space $(X,\|\cdot\|)$.

In the special case when $X$ is a Hilbert space with inner product $(\cdot,\cdot)$, $A$ is
a JBW-algebra called a spin factor, \cite[Def. 3.33]{AS1}, with the Jordan product
\[
(\alpha,y)\circ(\beta,z)=(\alpha\beta+(y,z), \alpha z+\beta y).
\]

\begin{theorem} Let $A$ be a generalized spin factor obtained from a Banach  space
$(X,\|\cdot\|)$. The following are equivalent.
\begin{enumerate}
\item $A$ is spectral and satisfies the condition \eqref{eq:JB_condition};
\item $A$ is a JB-algebra;
\item $A$ is a spin factor;
\item $A$ is a JBW-algebra.
\end{enumerate}

\end{theorem}

\begin{proof} By what was said above and Theorem \ref{th:JBalg1},  it is enough to show that if $A$ is a JB-algebra, then $A$ must be a spin factor.
So assume that $A$ is a JB-algebra. Then any sharp element $p\in A$ is a projection \cite[Lemma
5.7]{JP1}, and hence the focus of the compression $U_p$, moreover, the compressions
satisfy condition \eqref{eq:JB_condition}. Let $y\in X^*$, $\|y\|=1$. By \cite[Lemma 6.1]{JP1},
the element $p=1/2(1,y)\in A$ is sharp. But then $p$ must be the focus of a compression,
 hence by \cite[Prop. 6.4]{JP1}, $y$ must be an extremal point in the dual unit ball such that
 \[
\partial^*_y=\{x\in X,\ \|x\|=1,\ \<y,x\>=1\}\ne \emptyset.
 \]
 This implies that any nonzero element in $X^*$ attains its norm on the unit ball in $X$,
 so that $X$ is reflexive. Moreover, we also see that any boundary point of the dual unit
 ball is an extremal point, hence $X^*$ is strictly convex. Note that by \cite[Thm. 6.5]{JP1}, this
 shows that $A$ is a spectral order unit space.

For a sharp element $p=1/2(1,y)$
with $y\in X^*$, $\|y\|=1$, we have by \cite[Prop. 6.4]{JP1} that the compression with
focus $p$ must be of the form
\[
U_p((\alpha,w))=(\alpha +\<w,x_y\>)p,
\]
for some (fixed) choice of $x_y\in \partial_{y}^*$. Note also that we have $1-p=1/2(1, -y)$.
Let us check the condition
\eqref{eq:JB_condition} in this case. After some computations, we obtain
\begin{equation}\label{eq:JB_condition_factor}
x_{-y}=-x_y \ \text{ and } \ \<z,x_y\>=\<y,x_z\>,\qquad \forall y,z\in X^*,\
\|y\|=\|z\|=1.
\end{equation}
Let us define a mapping $\psi:X^*\to X$ as
\[
\psi(y)=\|y\|x_{\|y\|^{-1}y},\ y\ne 0,\quad \psi(0)=0.
\]
Then $\psi(ty)=t\psi(y)$ holds for all $t\in \mathbb R$. Further, for any $z,y,w\in X^*$
with $\|w\|=1$, we have by \eqref{eq:JB_condition_factor}
\begin{align*}
\<w,\psi(y+z)\>&=\|y+z\|\<w,x_{\|y+z\|^{-1}(y+z)}\>=\<y+z,x_w\>=\|y\|\<\frac{y}{\|y\|},x_w\>+\|z\|\<\frac{z}{\|z\|},x_w\>\\
&=\|y\|\<w,x_{\|y\|^{-1}y}\>+\|z\|\<w,x_{\|z\|^{-1}z}\>=\<w,\psi(y)+\psi(z)\>,
\end{align*}
so that $\psi$ is a linear map. Further, for $y,z\in X^*$,
\[
\<y,\psi(z)\>=\|z\|\<y,x_{\|z\|^{-1}z}\>=\|y\|\<z,x_{\|y\|^{-1}y}\>=\<z,\psi(y)\>
\]
and $\<y,\psi(y)\>=\|y\|^2$. It follows that $(y,z):= \<y,\psi(z)\>$ defines a
scalar product in $X^*$ such that $\|y\|=\sqrt{(y,y)}$. Hence $X^*$, and therefore also
$X=X^{**}$, is a Hilbert space, so that $A$ is a spin factor.

\end{proof}

\section{Rickart JB-algebras}

Let $A$ be a JB-algebra.
In \cite{Arz}, the following symbols were defined for a subset $S$ of $A$:

\begin{eqnarray}
& S^{\perp} = \{ a\in A: U_a(x)=0, \forall x\in S\},\\
& ^{\perp }S =\{x\in A: U_a(x)=0, \forall a\in A\},\\
& ^{\perp}S^+= {^\perp S}\bigcap A^+.
\end{eqnarray}

The following notion of a \emph{Rickart JB-algebra} was introduced by Ayupov and Arzikulov \cite{AyArz}.

\begin{definition} A JB-algebra $A$ is Rickart if one of the following equivalent statements is true
\begin{enumerate}
\item[(A1)] For every element $x\in A^+$ there is a projection $p\in A$ such that
\[
\{ x\} ^{\perp} =U_p(A)
\]
\item[(A2)] For  every element $x\in A$ there is a projection $p\in A$ such that
\[
^{\perp}\{ x\}^+ = U_p(A)^+.
\]
\end{enumerate}

\end{definition}

It was proved in \cite{AyArz} that the self-adjoint part of a C*-algebra $\mathcal A$ is a
Rickart JB-algebra if and only if  $\mathcal A$ is a Rickart C*-algebra, so that this
notion generalizes Rickart C*-algebras. The aim of this section is to prove the following
generalization of a result proved  for C*-algebras in  \cite{SW}.

\begin{theorem}\label{thm:Rickart} Let $A$ be a JB-algebra. Then $A$ is Rickart if and only if every maximal associative subalgebra in $A$ is monotone $\sigma$-complete.

\end{theorem}

The proof will be based on the following result.

\begin{theorem}\label{th:JBrickspect} {\rm \cite[Theorem 5.16]{JP1}} Let $A$ be a JB-algebra. Then $A$ is Rickart if and
only if $A$  is spectral.
\end{theorem}

We will need some further  preparations. For elements $a,b\in A$ we say that $a$ and $b$
\emph{operator commute} if $T_aT_b=T_bT_a$. If
one of the elements is a projection, then we have the following characterizations of
operator commutativity.

\begin{prop}\label{prop:opcom} {\rm \cite[Prop. 1.47]{AS1}} Let $a\in A$ and $p\in P$. Then the
following are equivalent.
\begin{enumerate}
\item $a$ and $p$ operator commute;
\item $T_p(a)=U_p(a)$;
\item $a=U_p(a)+U_{1-p}(a)$;
\item $a$ and $p$ are contained in an associative subalgebra.

\end{enumerate}

\end{prop}

The equivalence of (i) and (iv) was  generalized to all pairs of elements in \cite{Wet2}. Notice that by Lemma \ref{le:commut} (iii), (iii) is equivalent to $a\leftrightarrow p$.

\begin{lemma}\label{lemma:opcom_assoc} Let $A$ be a JB-algebra and let $C\subseteq A$ be a
subalgebra. Then  $C$ is associative if and only if the elements in $C$ mutually operator
commute.

\end{lemma}

\begin{proof} Assume that $C$ is associative and let $a,b\in C$. Then $a$ and $b$ generate
an associative JB-algebra, so that $a$ and $b$ operator commute, \cite[Thm. 3.13]{Wet2}.
Conversely, assume all elements in $C$ mutually operator commute, then for $a,b,c\in
C$, we have
\[
a\circ(b\circ c)=a\circ(c\circ b)=T_aT_c(b)=T_cT_a(b)=c\circ(a\circ b)=(a\circ b)\circ c.
\]

\end{proof}

For each element $a$ in a JB-algebra $A$, $A(a,1)$ denotes the norm-closed subalgebra generated by $a$ and $1$. By \cite[Corollary 1.4]{AS1}, and continuity of Jordan product,
$A(a,1)$ is associative.

\begin{theorem}\label{th:assoc} {\rm \cite[Proposition 1.12]{AS1}} If $A$ is a JB-algebra
and $B$ is a norm-closed associative subalgebra containing $1$, in particular if $B=A(a,1)$ for $a\in A$, then $B$ is isometrically (order- and algebra-) isomorphic to $C(X,{\mathbb R})$ for some compact Hausdorff space $X$.
\end{theorem}

Assume that $A$ has the comparability property. Recall that (cf. Definition \ref{de:extcompat})
\[
aCb \iff P(a)\leftrightarrow P(b).
\]
If $A$ is spectral, this is also equivalent to $P_{sp}(a)\leftrightarrow P_{sp}(b)$,
where $P_{sp}(a)$ denotes the Boolean subalgebra in $P$  generated by spectral projections of $a$
(Lemma \ref{le:spr}). If $b\in P$, then $aCb$ is equivalent to  $a\in C(b)$ and
Proposition \ref{prop:opcom} shows that for $b\in P$, $a\in C(b)$ is the same as
operator commutativity. We next show how this can be extended to arbitrary
$b\in A$.

\begin{lemma}\label{lemma:spectral_operator} Let $A$ be a JB-algebra with the
comparability property and let $a,b\in A$. Then  $aCb$ implies that $a$ and $b$ operator commute.

\end{lemma}

\begin{proof}
Assume $aCb$, so that $a\leftrightarrow P(b)$, which means that $a$ operator commutes with
all elements of $P(b)$. By Proposition \ref{prop:bicomutant} (iii),
 $b$ is a norm limit of a sequence of linear combinations of elements in $P(b)$.
Since  the map $b\mapsto T_b$ is linear and norm-continuous, we obtain that $a$ operator
commutes with $b$.

\end{proof}

\begin{corollary}\label{coro:cblock_jb} Assume that $A$ is a JB-algebra with the
comparability property. Then any C-block $C$ of $A$ is a maximal  associative subalgebra.

\end{corollary}

\begin{proof} Let $B\subseteq P$ be a block of $P$. We show that $C=C(B)$ is a subalgebra, which means that it is closed under
the Jordan product of its elements. Let $a\in C$, $p\in B$, then by Proposition
\ref{prop:opcom}, $p\circ a= T_pa=U_pa\in C$. This can be extended to all $a,b\in C$ using  Theorem
\ref{th:cor3.26}(iii), continuity of the map $b\mapsto T_b$ and the fact that $C$ is a
norm-closed linear subspace in $A$. Since we have
$aCb$ for all $a,b\in C$, we see from Lemmas \ref{lemma:spectral_operator} and
\ref{lemma:opcom_assoc} that $C$ is associative, hence it is contained in some maximal
associative subalgebra $C_0$. It follows that all elements of $C_0$ mutually operator
commute, which by Proposition \ref{prop:opcom} implies that all elements of $C_0$ commute
with all projections in $B$, so that  $C_0\subseteq C(B)=C$ and  $C_0=C$.

\end{proof}

Assume now that $A$ is Rickart, then each $a\in A$ has a \emph{carrier} $s(a)\in P$, which is  the
smallest projection such that $p\circ a=a$ \cite[Prop. 1.8]{AyArz}. It was shown in
\cite[Lemma 5.14]{JP1} that if $0\le a\le 1$, then $s(a)$ is the projection cover of $a$.
We also have the following relation to the Rickart mapping $a\mapsto a^*$ (see
\eqref{eq:rickart}).

\begin{lemma}\label{le:carrier} For every $a\in A$, $a^*=1-s(a)$.
\end{lemma}
\begin{proof} Let $a\in A$, $p\in P$. By (\ref{eq:UpTp}) we have
\[
p\circ a=0 \ \Leftrightarrow \ U_p(a)=0,  \qquad 
p\circ a=a \ \Leftrightarrow \ U_p(a)= a.
\]
As $p\circ a=a$ iff $(1-p)\circ a=0$ by distributivity of the Jordan product, in both the above cases we have $a\in C(p)$.

Put $p:=s(a)$. Then $p\circ a=a$ implies $a\in C(p)$ and $U_{1-p}(a)=0$. Let $q\in P$ be such that $a\in C(q)$ and $U_q(a)=0$. Then $(1-q)\circ a=a$, which entails $p\leq 1-q$.
Then $q\leq 1-p$, and by definition of the Rickart mapping, $a^*=1-p$.
\end{proof}

\begin{lemma}\label{le:support} Let $A$ be a Rickart JB-algebra and let $a\in A^+$. Then $s(a)\in A(a,1)$.
\end{lemma}

\begin{proof} By replacing $a$ with $\|a\|^{-1}a$, we may clearly assume that $0\le a \le
1$. Let $C$ be a C-block containing $a$. Since $A$ is spectral, $C\cong C(X)$
for a basically disconnected compact Hausdorff space $X$. Let $f\in C(X)$ be a function
corresponding to $a$ and let $Y\subset X$ be the support of $f$, then  $s(a)$ corresponds to the characteristic function $\chi_Y\in C(X)$. Then $0\le f(x)\le 1$  and $f_n(x):=f(x)^{1/n}$ is a nondecreasing sequence of continuous functions pointwise converging to $\chi_Y$. By Dini's theorem, $f_n$ converges to $\chi_Y$ in norm, so that the
 corresponding sequence $a_n:=a^{1/n}$ norm-converges to $s(a)$. Since all $a_n$ are contained in $A(a,1)$, so is $s(a)$.

\end{proof}

\begin{corollary}\label{coro:spectrum} Let $A$ be a Rickart JB-algebra and let $a\in A$, then $A(a,1)$ contains all the spectral projections of $a$.

\end{corollary}

\begin{proof} By Lemma \ref{le:carrier}, the spectral projections can be obtained as complements of supports of
elements of the form $(a-\lambda1)^+$, $\lambda \in \mathbb R$.
Since $(a-\lambda1)^+$ belong to $A(a,1)$ for all $\lambda$,
 all such supports belong to $A(a,1)$ by Lemma \ref{le:support}.

\end{proof}

\begin{lemma}\label{lemma:commut_rickart} Let $A$ be a Rickart JB-algebra. Then $aCb$ if and
only if $a$ and $b$ operator commute.

\end{lemma}

\begin{proof}
Assume that $A$ is Rickart. Note that if $b$ operator commutes with $a$ then it
operator commutes with $a^2$ \cite[Thm. 3.13]{Wet2} and hence with all polynomials in $a$
\cite[Cor. 2.11]{Wet2}. Consequently, $b$ operator commutes with all of $A(a,1)$. By Corollary \ref{coro:spectrum},
$A(a,1)$ includes all spectral projections of $a$, which implies that $aCb$. The converse
follows from Lemma \ref{lemma:spectral_operator} and the fact that $A$ is spectral.

\end{proof}

\begin{prop}\label{prop:Cblock_maxassoc} Let $A$ be a Rickart JB-algebra. Then the
C-blocks of $A$ are precisely the maximal associative subalgebras in $A$.

\end{prop}

\begin{proof} We already know by Corollary \ref{coro:cblock_jb} that any C-block is a
maximal associative subalgebra. Let  $C\subseteq A$ be a maximal associative subalgebra in
$A$. Then by Lemmas \ref{lemma:opcom_assoc} and \ref{lemma:commut_rickart}, all elements
of $C$ are mutually commuting so that $C$ is contained in some C-block $C_0$. By Corollary
\ref{coro:cblock_jb}, $C_0$ is a maximal associative subalgebra, hence $C=C_0$ is a
C-block.

\end{proof}

Finally, we will need the following characterization of a spectral order unit space,
proved in \cite{JP1}.

\begin{theorem}{\rm \cite{JP1}}\label{thm:spectral_complete} Let $A$ be a complete order unit space with a compression base having the comparability property. Then the following are equivalent.
\begin{enumerate}
\item $A$ is spectral;
\item Every C-block of $A$ is spectral;
\item Every C-block of $A$ is monotone $\sigma$-complete.
\end{enumerate}
\end{theorem}

\begin{proof}[Proof of Theorem \ref{thm:Rickart}]
Assume that $A$ is Rickart, so  the compression base $(U_p)_{p\in P}$ in $A$  is spectral.
Since $A$ is complete, we have by Theorem \ref{thm:spectral_complete} that
every C-block in $A$ is monotone $\sigma$-complete. The statement now follows by
Proposition  \ref{prop:Cblock_maxassoc}.

 Conversely, assume that every maximal associative subalgebra in $A$ is monotone
 $\sigma$-complete. We will first show that $A$ with the compression base $(U_p)_{p\in P}$ has the comparability property.
 Any element $a\in A$ is contained in some maximal associative subalgebra $A_0$, which is
 norm-closed by maximality. By Theorem \ref{th:assoc}, $A_0$ is isometrically isomorphic to
 $C(X)$ for some compact Hausdorff space $X$. By the assumption, $A_0$ is monotone
 $\sigma$-complete, so that $X$ is basically disconnected.
Let $f\in C(X)$ be the element corresponding to $a$ and
 let $f^+$ be the positive part, then $f^+=f\vee 0$ belongs to the subalgebra $A(f,1)$ in  $C(X)$
 generated by $f$ and 1. Similarly as in the proof of Lemma \ref{le:support},
 we obtain that the characteristic function of the support of $f^+$ is in $A(f,1)$. Hence
 there is a corresponding projection $p\in A(a,1)$.
  Observe that $A(a,1)\subseteq C(PC(a))$ (\cite[Eq. (10)]{JP1}), so that $p\in P(a)$
 and  it is easily seen by definition of $p$ that  $p\in P^\pm(a)$. Hence $A$ has comparability.

 Using Theorem \ref{thm:spectral_complete} once again, it is enough to show that any
 C-block $C$ of $A$ is a maximal associative subalgebra, which is precisely Corollary
 \ref{coro:cblock_jb}.

\end{proof}

\section{Concluding remarks}

We studied spectral order unit spaces in the sense of Foulis \cite{FP}. We proved that for
order unit spaces with comparability property,  a continuous functional calculus can be
introduced and there is a Borel functional calculus for spectral order unit spaces. 

We then concentrated on order unit spaces which are JB-algebras. As it turned out, the
condition of Alfsen and Shultz characterizing JB-algebras in the setting of spectral
duality \cite{AS1} can be applied to the much more general situation of order unit spaces
having the comparability property. In particular, in the case of order unit spaces derived
from Banach spaces (here we call them generalized spin factors), such an order unit space
is a JB-algebra if and only if it is a spin factor. 

Finally, we obtained a generalization of a characterization of Rickart C*-algebras due to
Sait\^o and Wright \cite{SW} to Rickart JB-algebras, using the connection  to  spectrality of
order unit spaces and recent results in \cite{Wet2}.

There is a number of problems related to spectrality of order unit spaces, which are left
to future work. For example, if $A$ is a JB-algebra, then $A\subseteq A^{**}$ and by \cite[Cor. 2.50]{AS1}, the second
dual is a JBW-algebra, hence it is spectral (in the sense of Alfsen and Shultz)
\cite[Thm. 2.20, Prop. 8.76]{AS1}, so that in particular,  any element $a\in A$ has a spectral resolution in
$A^{**}$. If $A$ is Rickart, then it follows from Theorem \ref{th:JBrickspect} that $a$ has a
spectral resolution in $A$. It is a question what is the  connection between these
spectral resolutions. Another question is the relation of the spectral resolutions of an
element $a$ and $g(a)$ for a Borel function $g$, here we may ask for conditions when the
integrals of the form \eqref{eq:gint} are defined, see the end of Section \ref{sec:fc}.

Another possible research direction is the study of convex effect algebras. 
The interval $[0,1]$ in an order unit space is an archimedean convex effect algebra
\cite{GPBB} and, conversely, any such effect algebra can be represented as a unit interval
in an order unit space. Spectrality in effect algebras was  studied in \cite{JP2} and it
was proved that an archimedean convex effect algebra is spectral (has the comparability
property) if and only if the
representing order unit space is spectral (has the comparability property). 
We may apply our results to some questions for convex effect algebras, for example,
for the study of convex sequential effect algebras, their spectrality and their relation
to JB-algebras.

\section*{Acknowledgements}

The research was supported by the grants VEGA 1/0142/20 and  the Slovak Research and
Development Agency grant APVV-20-0069.

\end{document}